\documentclass{article}
\usepackage{amsfonts}
\usepackage{graphicx}
\usepackage{amsmath,amssymb,amsthm}
\usepackage{authblk}
\usepackage{url}

\usepackage{microtype}
\newcommand{\REV}[1]{\ensuremath{\overline{#1}}}
\newcommand{\RLCSA}{\ensuremath{\mathsf{RLBWT}}}
\newcommand{\ST}{\ensuremath{\mathsf{ST}}}

\newcommand{\SA}{\ensuremath{\mathsf{SA}}}
\newcommand{\ISA}{\ensuremath{\mathsf{ISA}}}
\newcommand{\LCP}{\ensuremath{\mathsf{LCP}}}
\newcommand{\PLCP}{\ensuremath{\mathsf{PLCP}}}

\newcommand{\BWT}{\ensuremath{\mathsf{BWT}}}
\newcommand{\CDAWG}{\ensuremath{\mathsf{CDAWG}}}
\newcommand\SP[1]{\mathtt{sp}(#1)}
\newcommand\EP[1]{\mathtt{ep}(#1)} 
\newcommand\INTERVAL[1]{\mathtt{range}(#1)}

\newcommand{\newe}{e}

\newtheorem{definition}{Definition}
\newtheorem{lemma}{Lemma} 
\newtheorem{theorem}{Theorem}
\newtheorem{property}{Property}
\newtheorem{corollary}{Corollary}

\title{Representing the suffix tree with the CDAWG
}

\author[1]{Djamal Belazzougui}
\author[2]{Fabio Cunial}
\affil[1]{DTISI, CERIST Research Center, Algiers, Algeria.}
\affil[2]{Max Planck Institute of Molecular Cell Biology and Genetics, Dresden, Germany.}

\begin{document}

\maketitle

\begin{abstract}
Given a string $T$, it is known that its suffix tree can be represented using the compact directed acyclic word graph (CDAWG) with $\newe_T$ arcs, taking overall $O(\newe_T+\newe_{\REV{T}})$ words of space, where $\REV{T}$ is the reverse of $T$, and supporting some key operations in time between $O(1)$ and $O(\log{\log{n}})$ in the worst case. This representation is especially appealing for highly repetitive strings, like collections of similar genomes or of version-controlled documents, in which $e_T$ grows sublinearly in the length of $T$ in practice. In this paper we augment such representation, supporting a number of additional queries in worst-case time between $O(1)$ and $O(\log{n})$ in the RAM model, without increasing space complexity asymptotically. Our technique, based on a heavy path decomposition of the suffix tree, enables also a representation of the suffix array, of the inverse suffix array, and of $T$ itself, that takes $O(e_T)$ words of space, and that supports random access in $O(\log{n})$ time. Furthermore, we establish a connection between the reversed CDAWG of $T$ and a context-free grammar that produces $T$ and only $T$, which might have independent interest.
\end{abstract}

\section{Introduction}

Given a string $T$ of length $n$, the compressed suffix tree \cite{RNOtalg11,NRdcc14.1} and the compressed suffix array can take an amount of space that is bounded by the $k$-th order empirical entropy of $T$, but such measure of redundancy is known not to be meaningful when $T$ is very repetitive \cite{Gagie06a}, e.g. a collection of similar genomes. The space taken by such compressed data structures also includes a $o(n)$ term, typically $O(n/\mbox{polylog}(n))$, which can become an obstacle when $T$ is very compressible. Rather than compressing the suffix array, we could compress a \emph{differentially encoded suffix array} \cite{GNFjea14}, which stores at every position the difference between two consecutive positions of the suffix array. Previous approaches have compressed such differential array using grammar or Lempel-Ziv compression \cite{GNFjea14}, and the same methods can be used to compress the suffix tree topology and the LCP array \cite{ACNalgo13,NOsea14.1}. Such heuristics, however, have either no theoretical guarantee on their performance \cite{ACNalgo13,NOsea14.1}, or weak ones \cite{GNFjea14}.

In previous research \cite{belazzougui2015composite} we described a representation of the suffix tree of $T$ that takes space proportional to the size of the compact directed acyclic word graph (CDAWG) of $T$, and that supports a number of operations in time between $O(1)$ and $O(\log{\log{n}})$ in the worst case (see Table \ref{tab:suffixTree}). If $T$ is highly repetitive, the size of the CDAWG of $T$ is known to grow sublinearly in the length of $T$ in practice (see e.g. \cite{belazzougui2015composite}). Being related to maximal repeats, the size of the CDAWG is also a natural measure of redundancy for very repetitive strings. Moreover, since the difference between consecutive suffix array positions is the same inside isomorphic subtrees of the suffix tree, and since such isomorphic subtrees are compressed by the CDAWG, the CDAWG itself can be seen as a grammar that produces the differential suffix array, and the suffix tree can be seen as the parse tree of such grammar: this provides a formal substrate to heuristics that grammar-compress the differential suffix array.

In this paper we further exploit the compression of isomorphic subtrees of a suffix tree induced by the CDAWG, augmenting the representation of the suffix tree described in \cite{belazzougui2015composite} with a number of additional operations that take between $O(1)$ and $O(\log{n})$ time in the worst case (see Table \ref{tab:suffixTree2}), without increasing space complexity asymptotically. We also describe CDAWG-based representations of the suffix array, of the inverse suffix array, of the LCP array, and of $T$ itself, with $O(\log{n})$ random access time.

Our approach is related to the work of Bille et al~\cite{bille2015random}, 
in which a straight-line program (effectively a DAG) that produces the balanced parentheses representation of a tree with $n$ nodes, is used to support operations on the topology of the tree in $O(\log{n})$ time. Applying such compression to the suffix 
tree achieves the space bounds of this paper, but it only 
supports operations on the topology of the tree, and it supports each operation in $O(\log{n})$ time, whereas we achieve either constant or $O(\log{\log{n}})$ time for some key primitives. 

\begin{table}[t!]
\centering
\resizebox{\columnwidth}{!}{%
\begin{tabular}{|c|c|c|c|c|c|}
\hline
 & $\mathtt{leftmostLeaf}$ & $\mathtt{selectLeaf}$, $\mathtt{lca}$ & $\SA[i..j]$ & $T[i..j]$ & $\mathtt{depth}$ \\
 & $\mathtt{rightmostLeaf}$ & $\SA[i]$, $\ISA[i]$, $\LCP[i]$  & $\ISA[i..j]$ & & $\mathtt{ancestor}$ \\
 & & $\PLCP[i]$, $T[i]$ & $\LCP[i..j]$ & & $\mathtt{strAncestor}$ \\
\hline
1 & $O(1)$ & $O(\log{n})$ & $O(\log{n}+j-i)$ & $O(\log{n}+\frac{j-i}{\log_{\sigma}{n}})$ & $O(\log{n})$ \\
\hline
2 & $O(1)$ & $O(\log{n})$ & $O(\log{n}+j-i)$ & $O(\log{n}+\frac{j-i}{\log_{\sigma}{n}})$ & \\
\hline
\end{tabular}
}
\caption{Time complexity of the operations on the suffix tree of a string $T$ described in this paper ($n=|T|$).}
\label{tab:suffixTree2}
\end{table}

\begin{table}[b!]
\centering
\resizebox{\columnwidth}{!}{%
\begin{tabular}{|c|c|c|c|c|c|c|}
\hline
& Space & $\mathtt{stringDepth}$ & $\mathtt{isAncestor}$ & $\mathtt{parent}$ & $\mathtt{suffixLink}$ & $\mathtt{weinerLink}$ \\
& (words) & $\mathtt{nLeaves}$, $\mathtt{height}$ & $\mathtt{leafRank}$ & $\mathtt{nextSibling}$ & & \\ 
& & $\mathtt{locateLeaf}$ & & & & \\
& & $\mathtt{firstChild}$, $\mathtt{child}$ & & & & \\
\hline
1 & $O(\newe_T+\newe_{\REV{T}})$ & $O(1)$ & $O(1)$ & $O(\log{\log{n}})$ & $O(\log{\log{n}})$ & $O(\log{\log{n}})$ \\
\hline
2 & $O(\newe_T)$ & $O(1)$ & & $O(\log{\log{n}})$ & $O(1)$ & \\
\hline
\end{tabular}
}
\caption{Complexity of the operations on the suffix tree of a string $T$ described in \cite{belazzougui2015composite} ($n=|T|$).}
\label{tab:suffixTree}
\vspace{-1cm}
\end{table}

\section{Preliminaries}

We work in the RAM model with word length at least $\log{n}$ bits, where $n$ is the length of a string that is implicit from the context, and we index strings and arrays starting from one.

\subsection{Graphs}

We assume the reader to be familiar with the notions of tree and of directed acyclic graph (DAG). By $\mathtt{lca}(u,v)$ we denote the lowest common ancestor of nodes $u$ and $v$ in a tree. By \emph{weighted tree} we mean a tree with nonnegative weights on the edges, and we use $\omega(u,v)$ to denote the weight of edge $(u,v)$. Weighted DAGs are defined similarly. In this paper we only deal with \emph{ordered} trees and DAGs, in which there is a total order among the out-neighbors of every node. The $i$-th leaf of a tree is its $i$-th leaf in depth-first order, and to every node $v$ of a tree we assign the compact interval $[\SP{v}..\EP{v}]$, in depth-first order, of all leaves that belong to the subtree rooted at $v$. In this paper we use the expression DAG also for directed acyclic \emph{multigraphs}, allowing distinct arcs to have the same source and destination nodes. In what follows we consider just DAGs with exactly one source and one sink. 

We denote by $\mathcal{T}(G)$ the tree generated by DAG $G$ with the following recursive procedure: the tree generated by the sink of $G$ consists of a single node; the tree generated by a node $v$ of $G$ that is not the sink, consists of a node whose children are the roots of the subtrees generated by the out-neighbors of $v$ in $G$, taken in order, and connected to their parent by edges whose weight, if any, is identical to the weight of the corresponding arc of $G$. Note that: (1) every node of $\mathcal{T}(G)$ is generated by exactly one node of $G$; (2) a node of $G$ different from the sink generates one or more internal nodes of $\mathcal{T}(G)$, and the subtrees of $\mathcal{T}(G)$ rooted at all such nodes are isomorphic; (3) the sink of $G$ can generate one or more leaves of $\mathcal{T}(G)$; (4) there is a bijection, between the set of root-to-leaf paths in $\mathcal{T}(G)$ and the set of source-to-sink paths in $G$, such that every path $v_1,\dots,v_k$ in $\mathcal{T}(G)$ is mapped to a path $v'_1,\dots,v'_k$ in $G$, and such that $\omega(v_i,v_{i+1})=\omega(v'_i,v'_{i+1})$ for all $i \in [1..k-1]$ if $\mathcal{T}(G)$ is weighted. Symmetrically, given any tree $T$, merging all subtrees with identical topology and edge weights produces a DAG $G$ such that $\mathcal{T}(G)=T$: we denote such DAG by $\mathcal{G}(T)$. Clearly $\mathcal{G}(\mathcal{T}(G))=G$.

Given nodes $v$ and $w$ of $\mathcal{T}(G)$ such that $v$ is an ancestor of $w$, let $\mathtt{nLeaves}(v)$ be the number of leaves in the subtree rooted at $v$, and let $\mathtt{left}(v,w)$ (respectively, $\mathtt{right}(v,w)$) be the number of leaves in the subtree rooted at $v$ that precede (respectively, follow) in depth-first order the leaves in the subtree rooted at $w$. A \emph{heavy path decomposition} of $\mathcal{T}(G)$ \cite{harel1984fast} is the following marking: for every node $u$, we mark exactly one edge $(u,v)$ as \emph{heavy} if $\mathtt{nLeaves}(v)$ is the largest among all children of $u$, with ties broken arbitrarily (Figure \ref{fig:cdawg}a). We call \emph{light} an edge that is not heavy, and we call \emph{heavy path} a maximal sequence of nodes $v_1,\dots,v_k$ such that $(v_i,v_{i+1})$ is heavy for all $i \in [1..k-1]$. Note that $v_k$ is a leaf, every node of $\mathcal{T}(G)$ belongs to exactly one heavy path, distinct heavy paths are connected by light edges, and every path from the root to a leaf contains $O(\log{N})$ light edges, or equivalently intersects $O(\log{N})$ heavy paths, where $N$ is the number of leaves of $\mathcal{T}(G)$. Heavy paths are disjoint in $\mathcal{T}(G)$, but their corresponding paths in $G$ form a spanning tree $\tau(G)$, with $O(n)$ nodes and edges, rooted at the sink of $G$, where $n$ is the number of nodes of $G$ (Figure \ref{fig:cdawg}b).

\subsection{Strings}

Let $\Sigma=[1..\sigma]$ be an integer alphabet, let $\#=0 \notin \Sigma$ be a separator, and let $T \in [1..\sigma]^{n-1}\#$ be a string. Given a string $W \in [1..\sigma]^k$, we call the \emph{reverse of $W$} the string $\REV{W}$ obtained by reading $W$ from right to left. For a string $W \in [1..\sigma]^{k}\#$ we abuse notation, and we denote by $\REV{W}$ the string $\REV{W[1..k]}\#$. Given a substring $W$ of $T$, let $\mathcal{P}_{T}(W)$ be the set of all starting positions of $W$ in the circular version of $T$. A \emph{repeat} $W$ is a string that satisfies $|\mathcal{P}_{T}(W)|>1$. We denote by $\Sigma^{\ell}_{T}(W)$ the set of characters $\{a \in [0..\sigma] : |\mathcal{P}_{T}(aW)|>0\}$ and by $\Sigma^{r}_{T}(W)$ the set of characters $\{b \in [0..\sigma] : |\mathcal{P}_{T}(Wb)|>0\}$. A repeat $W$ is \emph{right-maximal} (respectively, \emph{left-maximal}) iff $|\Sigma^{r}_{T}(W)|>1$ (respectively, iff $|\Sigma^{\ell}_{T}(W)|>1$). It is well known that $T$ can have at most $n-1$ right-maximal repeats and at most $n-1$ left-maximal repeats. A \emph{maximal repeat} of $T$ is a repeat that is both left- and right-maximal. It is also well known that a maximal repeat $W \in [1..\sigma]^m$ of $T$ is the equivalence class of all the right-maximal strings $\{W[1..m],\dots,W[k..m]\}$ such that $W[k+1..m]$ is left-maximal, and $W[i..m]$ is not left-maximal for all $i \in [2..k]$.

For reasons of space we assume the reader to be familiar with the notion of \emph{suffix tree} $\ST_T$ of $T$ (see e.g. \cite{gusfield1997algorithms} for an introduction), which we do not define here. We denote by $\ell(\gamma)$, or equivalently by $\ell(u,v)$, the string label of edge $\gamma=(u,v) \in E$, and we denote by $\ell(v)$ the string label of node $v \in V$. It is well known that a substring $W$ of $T$ is right-maximal iff $W=\ell(v)$ for some internal node $v$ of the suffix tree. We assume the reader to be familiar with the notion of \emph{suffix link} connecting a node $v$ with $\ell(v)=aW$ for some $a \in [0..\sigma]$ to a node $w$ with $\ell(w)=W$. Here we just recall that inverting the direction of all suffix links yields the so-called \emph{explicit Weiner links}. 

Finally, we assume the reader to be familiar with the notion and uses of the Burrows-Wheeler transform of $T$ (see e.g. \cite{ferragina2005indexing}). In this paper we use $\BWT_T$ to denote the BWT of $T$, and we use $\INTERVAL{W} = [\SP{W}..\EP{W}]$ to denote the lexicographic interval of a string $W$ in a BWT that is implicit from the context. As customary, we denote by $C[0..\sigma]$ the array such that $C[a]$ equals the number of occurrences of characters lexicographically smaller than $a$ in $T$. For a node $v$ of $\ST_T$, we use the shortcut $\INTERVAL{v}=[\SP{v}..\EP{v}]$ to denote $\INTERVAL{\ell(v)}$. We say that $\BWT_{T}[i..j]$ is a \emph{run} iff $\BWT_{T}[k]=c \in [0..\sigma]$ for all $k \in [i..j]$, and moreover if any substring $\BWT_{T}[i'..j']$ such that $i' \leq i$, $j' \geq j$, and either $i' \neq i$ or $j' \neq j$, contains at least two distinct characters. We denote by $\mathcal{R}_{T}$ the set of all triplets $(c,i,j)$ such that $\BWT_{T}[i..j]$ is a run of character $c$. Given a string $T \in [1..\sigma]^{n-1}\#$, we call \emph{run-length encoded BWT} ($\RLCSA_T$) any representation of $\BWT_T$ that takes $O(|\mathcal{R}_T|)$ words of space, and that supports the well known rank and select operations: see for example \cite{makinen2005succinct1,MakinenNSV10,SirenVMN08}. It is easy to implement a version of $\RLCSA_T$ that supports rank in $O(\log{\log{n}})$ time and select in $O(\log{\log{n}})$ time~\cite{belazzougui2015composite}.

\subsection{CDAWG}

\begin{figure}[t!]
\centering
\includegraphics[width=0.94\textwidth]{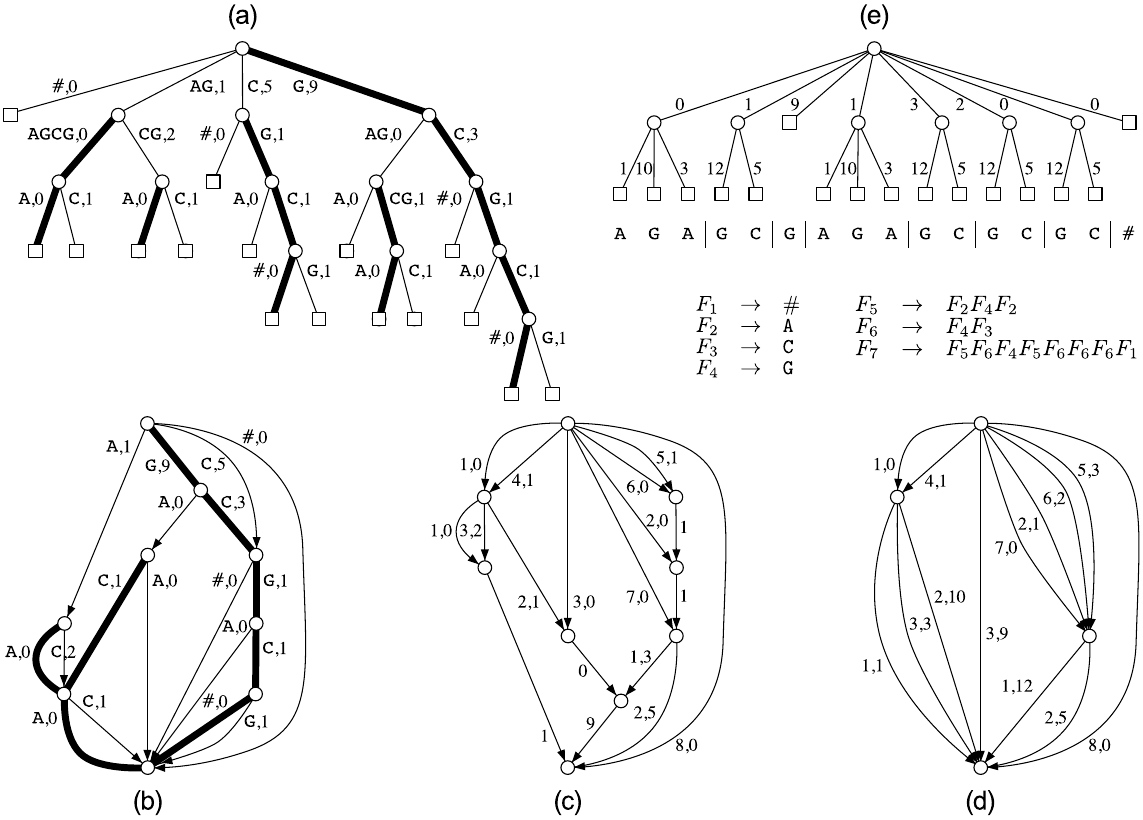}
\caption{The data structures used in this paper for string $T=\mathtt{AGAGCGAGAGCGCGC\#}$. (a) The suffix tree of $T$. Edges to leaves are labelled by just the first character of their string. The weight of edge $(u,v)$ is $\SP{v}-\SP{u}$. Heavy edges according to the number of leaves are bold. (b) The CDAWG of $T$. Just the first character of each arc label is shown. Arc weights are from (a). Arcs in the spanning tree $\tau$ are bold. (c) The reverse CDAWG. Arc $(u,v)$ is labelled by pair $(x,y)$, where $x$ is the order of $v$ among the out-neighbors of $u$, and $y$ is the weight in (b). (d) The compacted version of (c). (e) The weighted tree generated from (d), and the corresponding grammar.}
\label{fig:cdawg}
\end{figure}

The \emph{compact directed acyclic word graph} of a string $T$ (denoted by $\CDAWG_T$ in what follows) is the minimal compact automaton that recognizes the suffixes of $T$ \cite{blumer1987complete,CrochemoreV97}. We denote by $\newe_T$ the number of arcs in $\CDAWG_T$. The CDAWG of $T$ can be seen as the minimization of $\ST_T$, in which all leaves are merged to the same node (the sink) that represents $T$ itself, and in which all nodes except the sink are in one-to-one correspondence with the maximal repeats of $T$ \cite{Raffinot2001}. Every arc of $\CDAWG_T$ is labeled by a substring of $T$, and the out-neighbors $w_1,\dots,w_k$ of every node $v$ of $\CDAWG_T$ are sorted according to the lexicographic order of the distinct labels of arcs $(v,w_1),\dots,(v,w_k)$. Since there is a bijection between the nodes of $\CDAWG_T$ and the maximal repeats of $T$, the node $v'$ of $\CDAWG_T$ with $\ell(v')=W$ is the equivalence class of the nodes $\{v_1,\dots,v_k\}$ of $\ST_T$ such that $\ell(v_i)=W[i..|W|]$ for all $i \in [1..k]$, and such that $v_k,v_{k-1},\dots,v_1$ is a maximal unary path of explicit Weiner links. The subtrees of $\ST_T$ rooted at all such nodes are isomorphic, and $\mathcal{T}(\CDAWG_T)=\ST_T$ (Figure \ref{fig:cdawg}b).
It follows that the set of right-maximal strings that belong to the equivalence class of a maximal repeat can be represented by a single integer $k$, and a right-maximal string can be identified by the maximal repeat $W$ it belongs to, and by the length of the corresponding suffix of $W$. Similarly, a suffix of $T$ can be identified by a length relative to the sink of $\CDAWG_T$.

In $\BWT_T$, the right-maximal strings in the same equivalence class of a maximal repeat enjoy the following properties:

\begin{property}[\cite{belazzougui2015composite}]\label{obs:equivalenceClassInBWT}
Let $\{W[1..m],\dots,W[k..m]\}$ be the right-maximal strings that belong to the equivalence class of maximal repeat $W \in [1..\sigma]^m$ of a string $T$, and let $\INTERVAL{W[i..m]}=[p_i..q_i]$ for $i \in [1..k]$. Then: (1) $|q_i-p_i+1|=|q_j-p_j+1|$ for all $i$ and $j$ in $[1..k]$; (2) $\BWT_{T}[p_i..q_i]=W[i-1]^{q_i-p_i+1}$ for $i \in [2..k]$. Conversely, $\BWT_{T}[p_1..q_1]$ contains at least two distinct characters. (3) $p_{i-1} = C[c]+\mathtt{rank}_{c}(\BWT_T,p_i)$ and $q_{i-1}=p_{i-1}+q_i-p_i$ for $i \in [2..k]$, where $c=W[i-1]=\BWT_{T}[p_i]$. (4) $p_{i+1} = \mathtt{select}_{c}(\BWT_T,p_i-C[c])$ and $q_{i+1} = p_{i+1}+q_i-p_i$ for $i \in [1..k-1]$, where $c=W[i]$ is the character that satisfies $C[c] < p_i \leq C[c+1]$. (5) Let $c \in [0..\sigma]$, and let $\INTERVAL{W[i..m]c}=[x_i..y_i]$ for $i \in [1..k]$. Then, $x_i = p_i+x_1-p_1$ and $y_i = p_i+y_1-p_1$.
\end{property}

Character $c$ in Property \ref{obs:equivalenceClassInBWT}.4 can be computed in $O(\log{\log{n}})$ time using a predecessor data structure that uses $O(\sigma)$ words of space \cite{Wi83}. Moreover, the equivalence class of a maximal repeat is related to the equivalence classes of its in-neighbors in the CDAWG in the following way:

\begin{property}[\cite{belazzougui2015composite}]\label{obs:inNeighbors}
Let $w$ be a node in $\CDAWG_T$ with $\ell(w) = W \in [1..\sigma]^m$, and let $\mathcal{S}_w=\{W[1..m]$, $\dots$, $W[k..m]\}$ be the right-maximal strings that belong to the equivalence class of node $w$. Let $\{v^1,\dots,v^t\}$ be the in-neighbors of $w$ in $\CDAWG_T$, and let $\{V^1,\dots,V^t\}$ be their labels. Then, $\mathcal{S}_w$ is partitioned into $t$ disjoint sets $\mathcal{S}_w^1,\dots,\mathcal{S}_w^t$ such that $\mathcal{S}_w^i = \{W[x^i+1..m],W[x^i+2..m],\dots,W[x^i+|\mathcal{S}_{v^i}|..m]\}$, and the right-maximal string $V^{i}[p..|V^i|]$ labels the parent of the locus of the right-maximal string $W[x^i+p-1..m]$ in $\ST_T$.
\end{property}

Property \ref{obs:inNeighbors} applied to the sink $v$ of $\CDAWG_T$ partitions $T$ into $x$ left-maximal factors, where $x$ is the number of in-neighbors of $v$ (Figure \ref{fig:cdawg}e). Moreover, by Property \ref{obs:inNeighbors}, it is natural to say that in-neighbor $v^i$ of node $w$ is smaller than in-neighbor $v^j$ of node $w$ iff $x^i<x^j$, or equivalently if the strings in $\mathcal{S}^i_w$ are longer than the strings in $\mathcal{S}^j_w$. We call $\REV{\CDAWG}_T$ the ordered DAG obtained by applying this order to the reverse of $\CDAWG_T$, i.e. to the DAG obtained by inverting the direction of all arcs of $\CDAWG_T$ (Figure \ref{fig:cdawg}c). Note that $\REV{\CDAWG}_T$ is not the same as $\CDAWG_{\REV{T}}$, although there is a bijection between their sets of nodes. Note also that some nodes of $\REV{\CDAWG}_T$ can have just one out-neighbor: for brevity we denote by $\REV{\CDAWG}_T$ the graph obtained by collapsing every such node $v$, i.e. by adding the weight (if any) of the only outgoing arc from $v$ to the weights of all incoming arcs to $v$, and by redirecting such incoming arcs to the out-neighbor of $v$ (Figure \ref{fig:cdawg}d). This can be done in linear time by an inverse topological sort of $\REV{\CDAWG}_T$ that starts from its sink.

The source of $\REV{\CDAWG}_T$ is the sink of $\CDAWG_T$, which is the equivalence class of all suffixes of $T$ in string order, and there is a bijection between the distinct paths of $\REV{\CDAWG}_T$ and the suffixes of $T$. It follows that:

\begin{property}
The $i$-th leaf of $\mathcal{T}(\REV{\CDAWG}_T)$ in depth-first order corresponds to the $i$-th suffix of $T$ in string order.
\end{property}

Thus, $\mathcal{T}(\REV{\CDAWG}_T)$ can be seen as the parse tree of a context-free grammar that generates $T$ and only $T$, and $\REV{\CDAWG}_T$ can be seen as such grammar (Figure \ref{fig:cdawg}e). This implies a lower bound on the size of the CDAWG:

\begin{lemma}
Let $f$ be the function that maps the length of a string to the size of its CDAWG, and let $g$ be the function that maps the length of a string $T$ to the size of the smallest grammar that produces $T$ and only $T$. Then, $f \in \Omega(g)$.
\end{lemma}

In some classes of strings the size of the CDAWG is asymptotically the same as the size of the smallest grammar that produces the string, but in other classes the ratio between the two sizes reaches its maximum, $O(n/\log{n})$: see Section 2.1 in \cite{belazzougui2015composite}. 

Let $G$ be an ordered DAG, let $\gamma=(v,w)$ be an edge of $\mathcal{T}(G)$, and assume that we assign to $\gamma$ a weight equal to the offset $\SP{w}-\SP{v}$ between the first leaf in the leaf interval of $w$ and the first leaf in the leaf interval of $v$ (Figure \ref{fig:cdawg}a). Thus, we can compute the depth-first order of a leaf of $\mathcal{T}(G)$ by summing the weights of all edges in its root-to-leaf path. Note that edges $(v,w)$ and $(v',w')$ in $\mathcal{T}$ such that $v$ and $v'$ correspond to the same node $v''$ in $G$, and such that $w$ and $w'$ correspond to the same node $w''$ in $G$, have the same weight: in the case of $\CDAWG_T$ and $\ST_T$, this is equivalent to Property \ref{obs:equivalenceClassInBWT}.5, and weights are offsets between the starting positions of nested BWT intervals (Figure \ref{fig:cdawg}b). Assume that every such weight is stored inside arc $(v'',w'')$ of $\CDAWG_T$, and that weights are preserved when building $\REV{\CDAWG}_T$. Then, one plus the sum of all weights in the source-to-sink path of $\REV{\CDAWG}_T$ that corresponds to suffix $T[i..|T|]$ is the lexicographic rank of suffix $T[i..|T|]$ (see e.g. Figures \ref{fig:cdawg}d and \ref{fig:cdawg}e). Equivalently:

\begin{property}\label{property:revCDAWG}
Let arc $(u,v)$ of $\CDAWG_T$ be weighted by $\SP{v'}-\SP{u'}$, where $v'$ (respectively, $u'$) is a node of $\ST_T$ that belongs to the equivalence class of $v$ (respectively, $u$), and $v'$ is a child of $u'$ in $\ST_T$. Then, the lexicographic rank of suffix $T[i..|T|]$ is one plus the sum of all weights in the path from the root of $\mathcal{T}(\REV{\CDAWG}_T)$ to the $i$-th leaf of $\mathcal{T}(\REV{\CDAWG}_T)$ in depth-first order.
\end{property}

\subsection{Representing the suffix tree with the CDAWG}\label{sec:cpm2015}

It is known that Properties \ref{obs:equivalenceClassInBWT} and \ref{obs:inNeighbors} enable two encodings of $\ST_T$ that take $O(\newe_T+\newe_{\REV{T}})$ words of space each, and that support the operations in Table \ref{tab:suffixTree} with the specified time complexities \cite{belazzougui2015composite}. Since the rest of this paper builds on the representation described in \cite{belazzougui2015composite}, we summarize it here for completeness.

It is known that $|\mathcal{R}_T|$ is at most the number of arcs in $\CDAWG_T$ \cite{belazzougui2015composite}, thus augmenting $\CDAWG_T$ with $\RLCSA_T$ does not increase space asymptotically. For every node $v$ of $\CDAWG_T$, we store: $|\ell(v)|$ in a variable $v\mathtt{.length}$; the number $v\mathtt{.size}$ of right-maximal strings that belong to its equivalence class; the interval $[v\mathtt{.first}..v\mathtt{.last}]$ of $\ell(v)$ in $\BWT_T$; a linear-space predecessor data structure \cite{Wi83} on the boundaries induced on the equivalence class of $v$ by its in-neighbors (Property \ref{obs:inNeighbors}); and pointers to the in-neighbor that corresponds to the interval associated with each boundary. For every arc $\gamma=(v,w)$ of $\CDAWG_T$, we store the first character of $\ell(\gamma)$ in a variable $\gamma\mathtt{.char}$, and the number of characters of the right-extension implied by $\gamma$ in a variable $\gamma\mathtt{.right}$. We also add to the CDAWG all arcs $(v,w,c)$ such that $w$ is the equivalence class of the destination of a Weiner link from $v$ labeled by character $c$ in $\ST_T$, and the reverse of all explicit Weiner link arcs. We represent a node $v$ of $\ST_T$ as a tuple $\mathtt{id}(v)=(v',|\ell(v)|,i,j)$, where $v'$ is the node in $\CDAWG_T$ that corresponds to the equivalence class of $v$, and $[i..j]$ is the interval of $\ell(v)$ in $\BWT_T$. Implementing operations $\mathtt{stringDepth}(\mathtt{id}(v))$, $\mathtt{nLeaves}(\mathtt{id}(v))$ (which returns the number of leaves of the subtree of $\ST_T$ rooted at a given node), $\mathtt{isAncestor}(\mathtt{id}(v),\mathtt{id}(w))$ (which returns true iff a node $v$ of $\ST_T$ is an ancestor of another node $w$ of $\ST_T$), $\mathtt{suffixLink}(\mathtt{id}(v))$, $\mathtt{weinerLink}(\mathtt{id}(v))$, $\mathtt{locateLeaf}(\mathtt{id}(v))$ (which returns the position in $T$ of a leaf $v$ of $\ST_T$) and $\mathtt{leafRank}(\mathtt{id}(v))$ (which returns the position of a leaf $v$ of $\ST_T$ in lexicographic order) is straightforward using Properties \ref{obs:equivalenceClassInBWT}.3 and \ref{obs:equivalenceClassInBWT}.4, and implementing $\mathtt{parent}(\mathtt{id}(v))$, $\mathtt{child}(\mathtt{id}(v))$ and $\mathtt{nextSibling}(\mathtt{id}(v))$ is easy using Properties \ref{obs:inNeighbors} and \ref{obs:equivalenceClassInBWT}.5. 

Removing all implicit Weiner link arcs from our data structure achieves $O(\newe_T)$ words of space, and still supports all queries except following implicit Weiner links. We can further drop $\RLCSA_T$ and remove from $\mathtt{id}(v)$ the interval of $\ell(v)$ in $\BWT_T$, still supporting most of the original queries in the same amount of time, and $\mathtt{suffixLink}$ in constant time. The data structure after such removals corresponds to the second row of Table \ref{tab:suffixTree}. Conversely, storing also the RLBWT of $\REV{T}$, and the interval in such RLBWT of the reverse of the maximal repeat that corresponds to every node of the CDAWG, allows one to also read the label of an edge $\gamma$ of $\ST_T$ in $O(\log{\log{n}})$ time per character, for the same asymptotic space complexity.

\section{Additional suffix tree operations}

In this paper we augment the representation of the suffix tree described in Section \ref{sec:cpm2015}, enabling it to support a number of additional suffix tree operations in $O(\log{n})$ time without increasing space complexity asymptotically. At the core of our methods lies a heavy path decomposition of $\CDAWG_T$ along the lines of \cite{bille2015random}, which we summarize in what follows to keep the paper self-contained.

\begin{definition}[Smooth function]
Let $T$ be a tree, let $v_1,v_2,\dots,v_N$ be its $N$ leaves in depth-first order, let $f$ be a function that assigns a real number to every leaf, and let $F[1..N]$ be the array that stores at position $i$ the value of $f(v_i)$. We say that $f$ is \emph{smooth} \emph{with respect to} $T$ iff $F[\SP{v}..\EP{v}]=F[\SP{w}..\EP{w}]$ for every pair of internal nodes $v,w$ of $T$ that are generated by the same node of $\mathcal{G}(T)$.
\end{definition}

For example, let $T$ be the parse tree of a string $S$ generated by a context-free grammar: the function that assigns character $T[i]$ to every position $i$ of $T$ is smooth.

\begin{lemma}[\cite{bille2015random}]\label{lemma:bille}
Let $G$ be a DAG with $n$ arcs such that every node has exactly two out-neighbors, let $f$ be a smooth function with respect to $\mathcal{T}(G)$, and let $N$ be the number of leaves of $\mathcal{T}(G)$. There is a data structure that, given a number $i \in [1..N]$, returns $f(u_i)$ in $O(\log{N})$ time, where $u_i$ is the $i$-th leaf of $\mathcal{T}(G)$ in depth-first order. Moreover, given two integers $1 \leq i \leq j \leq N$, the data structure returns in $O(\log{N})$ time the node of $G$ that corresponds to $\mathtt{lca}(u_i,u_j)$, and it returns in $O(\log{N}+j-i)$ time the sequence of values $f(u_i),f(u_{i+1}),\dots,f(u_j)$, where $u_h$ is the $h$-th leaf of $\mathcal{T}(G)$ in depth-first order. Such data structure takes $O(n)$ words of space. 
\end{lemma}
\begin{proof}[Proof sketch]
For each heavy path $v_1,\dots,v_k$ of $\mathcal{T}(G)$, we store at $v_1$ values $\mathtt{nLeaves}(v_1)$, $\mathtt{left}(v_1,v_k)$, $f(v_k)$, a predecessor data structure on the set of values $\{\mathtt{left}(v_1,v_i) : i \in [2..k]\}$, and a predecessor data structure on the set of values $\{\mathtt{right}(v_1,v_i) : i \in [2..k]\}$. If we query $v_1$ with the position $i_1$ of a leaf in the subtree rooted at $v_1$, such data structures allow us to detect the largest $j \in [1..k]$ such that $v_j$ is an ancestor of the query leaf. If $j=k$ we return $f(v_k)$, otherwise we take the light edge $(v_j,w)$ and we recur on $w$, which is itself the first node of a heavy path. This solution takes $O(\log{N})$ queries to prefix-sum data structures, but the total size of all prefix-sum data structures can be $O(N^2)$.

Note that a predecessor query on the left and right predecessor data structures stored at the first node $v_1$ of a heavy path of $\mathcal{T}(G)$ can be implemented with a \emph{weighted ancestor} query\footnote{A \emph{weighted ancestor query} $(v,k)$ on a tree with weights on the edges asks for the lowest ancestor $u$ of a node $v$ such that the sum of weights in the path from $u$ to $v$ is at least $k$ \cite{Amir:2007:DTS:1240233.1240242}.} on $\tau(G)$, if we assign to each arc $(v,w)$ of $G$ that also belongs to $\tau(G)$ a left weight equal to zero if $w$ is the left successor of $v$, and equal to the number of leaves in the left successor of $v$ otherwise (the right weight is defined similarly). Using a suitable data structure for weighted ancestor queries allows one to achieve $O(n)$ words of space and overall $O(\log{N} \cdot \log{\log{N}})$ query time after $O(n)$ preprocessing of $G$. More advanced data structures that implement weighted ancestor queries on $\tau(G)$ allow one to achieve the claimed bounds \cite{bille2015random}.

Given $\mathcal{T}(G)$, we proceed as follows to extract the values of all leaves in a depth-first interval $[i..j]$. Inside every node $v$ of a heavy path, we store an auxiliary right pointer to the closest descendant of $v$ in the heavy path whose right child is light. We symmetrically store an auxiliary left pointer. Then, we traverse $\mathcal{T}(G)$ top-down as described above, but searching for both the $i$-th leaf $u_i$ and the $j$-th leaf $u_j$ at the same time: when the nodes $w$ and $w'$ of $G$ that result from such searches are different, we know that one is a descendant of the other in $\tau(G)$, and the node of $G$ that corresponds to $\mathtt{lca}(u_i,u_j)$ in $\mathcal{T}(G)$ is the one whose number of leaves equals $\max\{\mathtt{nLeaves}(w),\mathtt{nLeaves}(w')\}$. Then we continue the search for the two leaves separately: during the search for $u_i$ (respectively, $u_j$) we follow all right (respectively, left) auxiliary pointers in all heavy paths, and we concatenate the corresponding nodes in a left (respectively, right) linked list. The size of such lists is $O(j-i)$, and computing sequence $f(u_i),\dots,f(u_j)$ from the lists takes $O(j-i)$ time. The same approach can be applied to $G$, at the cost of $O(n)$ preprocessing time and space.
\end{proof}

Since a node $v$ of $\mathcal{T}(G)$ can be uniquely identified by an interval of leaves in depth-first order, Lemma \ref{lemma:bille} effectively implements a map from the identifier of a node in $\mathcal{T}(G)$ to the identifier of its corresponding node in $G$. 

\begin{lemma}\label{lemma:expansion}
Lemma \ref{lemma:bille} holds also for a DAG in which all nodes have out-degree \emph{at least} two.
\end{lemma}
\begin{proof}
We expand every node $v$ with out-degree $d>2$ into a binary directed tree, with $d-1$ artificial internal nodes, whose $d$ leaves are the out-neighbors of $v$ in $G$. We also store in each artificial internal node $w$ a pointer $w.\mathtt{real}=v$. The size of such expanded DAG $G'$ is still $O(n)$, where $n$ is the number of arcs of $G$, $\mathcal{T}(G')$ is a binary tree with the same number of leaves as $\mathcal{T}(G)$, there is a bijection between the leaves of $\mathcal{T}(G)$ and the leaves of $\mathcal{T}(G')$ such that the $i$-th leaf in depth-first order in $\mathcal{T}(G)$ corresponds to the $i$-th leaf in depth-first order in $\mathcal{T}(G')$, and the extension of $f$ to the leaves of $\mathcal{T}(G')$ induced by such bijection is smooth with respect to $\mathcal{T}(G')$. Note that, if Lemma \ref{lemma:bille} returns an artificial node $w$ as the result of a lowest common ancestor query, it suffices to return $w.\mathtt{real}$ instead.
\end{proof}

Lemma \ref{lemma:bille} can be adapted to support queries on another class of functions:

\begin{definition}[Telescoping function]
Let $f$ be a function that assigns a real number to any path of any weighted graph. We say that $f$ is \emph{telescoping} iff:
\begin{enumerate}
\item Given a path $P=v_1,v_2,\dots,v_k$, $f(P) = g(\omega(v_1,v_2)) \circ \cdots \circ g(\omega(v_{k-1},v_k))$, where $\omega(v_i,v_j)$ is the weight of edge or arc $(v_i,v_j)$, $g$ is a function that can be computed in constant time, and $x \circ y$ is a binary associative operator with identity element $\mathbb{I}$ that can be computed in constant time.
\item $f(v_1,\dots,v_k) \geq f(v_1,\dots,v_{i})$ for all $i<k$, and $f(v_1,\dots,v_k) \geq f(v_i,\dots,v_k)$ for all $i>1$.
\item For every path $v_1,\dots,v_i,\dots,v_j,\dots,v_k$, $f(v_i,\dots,v_j)$ can be computed in constant time given $f(v_1,\dots,v_i)$ and $f(v_1,\dots,v_j)$, or given $f(v_i,\dots,v_k)$ and $f(v_j,\dots,v_k)$.
\end{enumerate}
\end{definition}

We call $y$ the \emph{inverse} of $x$ with respect to $\circ$ iff $x \circ y = y \circ x = \mathbb{I}$. For example, the sum of edge weights in a path is telescoping, $\mathbb{I}=0$, and the inverse of $x$ is $-x$. Note that a telescoping function is not necessarily smooth.

\begin{lemma}\label{lemma:DAG_sum_weights}
Let $G$ be a weighted DAG with $n$ arcs in which every node has at least two out-neighbors, let $f$ be a telescoping function, and let $N$ be the number of leaves of $\mathcal{T}(G)$. There is a data structure that, given a number $i \in [1..N]$, evaluates $f$ in $O(\log{N})$ time on the path from the root of $\mathcal{T}(G)$ to the $i$-th leaf in depth-first order. Moreover, given two numbers $1 \leq i \leq j \leq N$, the data structure: 
\begin{enumerate}
\item Evaluates $f$ in $O(\log{N})$ time on the path from the root of $\mathcal{T}(G)$ to $\mathtt{lca}(u_i,u_j)$, where $u_i$ and $u_j$ are the $i$-th and $j$-th leaf of $\mathcal{T}(G)$ in depth-first order.
\item Returns in $O(\log{N}+j-i)$ time the sequence of values $f(u_i),f(u_{i+1}),\dots,f(u_j)$, where $f(u_h)$ is the value of function $f$ evaluated on the path from the root of $\mathcal{T}(G)$ to the $h$-th leaf in depth-first order.
\item If $[i..j]$ is the identifier of node $v$ in $\mathcal{T}(G)$, given a nonnegative number $k$, returns in $O(\log{N})$ time the node of $G$ that corresponds to the highest ancestor $w$ of $v$ in $\mathcal{T}(G)$ such that $f$, evaluated on the path from the root of $\mathcal{T}(G)$ to $w$, is at least $k$ \emph{(weighted ancestor query)}.
\end{enumerate}
Such data structure takes $O(n)$ words of space. 
\end{lemma}
\begin{proof}
If a node $v$ in the DAG has out-degree greater than two, we expand it as described in Lemma \ref{lemma:expansion}, assigning weight $\mathbb{I}$ to all arcs that end in an artificial internal node of the expanded DAG, and assigning the weight of arc $(v,w)$ to the arc that connects an artificial internal node to out-neighbor $w$ of $v$ in $G$. We also store a pointer to $v$ inside each artificial internal node. Let $G'$ be the expanded version of $G$. At every node $v$ of $G'$ we store variable $v.\mathtt{count}=f(P(v))$, where $P(v)$ is the path from $v$ to the sink of $G'$ that uses only arcs in the spanning tree $\tau(G')$. We traverse $G'$ as described in Lemma \ref{lemma:bille}: at the current node $u$, we compute its highest ancestor $v$ in $\tau(G')$ that lies in the path, from the source of $G'$ to the sink of $G'$, that corresponds to the $i$-th leaf of $\mathcal{T}(G')$. We use $u.\mathtt{count}$ and $v.\mathtt{count}$ to evaluate $f$ in constant time on the path from $u$ to $v$ along $\tau(G')$, and we cumulate such value to the output. For each arc $(v,w)$ that does not belong to $\tau(G')$, we compute $g(\omega(v,w))$ and we cumulate it to the output. 

To evaluate $f$ on the path from the root of $\mathcal{T}(G)$ to $\mathtt{lca}(v_i,v_j)$, we follow the extraction strategy described in Lemma \ref{lemma:bille}, using in the last step $u.\mathtt{count}$ and $v.\mathtt{count}$, where $u$ is the current node and $v$ is the (possibly artificial) node of $G'$ that corresponds to $\mathtt{lca}(v_i,v_j)$ in $\mathcal{T}(G')$. We use the extraction strategy of Lemma \ref{lemma:bille} also to evaluate $f$ on all leaves of $\mathcal{T}(G)$ in the depth-first interval $[i..j]$: every time we take a right pointer or a left pointer $(u,v)$, we cumulate weight $u.\mathtt{count} \circ y$ to the current value of $f$, where $y$ is the inverse of $v.\mathtt{count}$, and we start from such value of $f$ when visiting the subgraph of $G'$ that starts at $v$.

To support weighted ancestor queries on $f$ and $\mathcal{T}(G)$, we build a data structure that supports \emph{level ancestor queries} on $\tau(G')$: given a node $v$ and a path length $d$, such data structure returns the ancestor $u$ of $v$ in $\tau(G')$ such that the path from the root of $\tau(G')$ to $u$ contains exactly $d$ nodes. The level ancestor data structure described in \cite{bender2004level,berkman1994finding} takes $O(n)$ words of space and it answers queries in constant time. 
We search again for the $i$-th and $j$-th leaf in parallel, cumulating $f$ using the weights of light arcs and of heavy paths as done before. Let $u$ be the current node in this search, and let $x$ be the current value of $f$: if $x<k$, but the value of $f$ is at least $k$ at the next node $v$ such that the path from $u$ to $v$ in $G'$ belongs to $\tau(G')$, we binary search the nodes $w$ on the path from $u$ to $v$, using level ancestor queries from $u$ and comparing $x \circ u.\mathtt{count} \circ y$ to $k$, where $y$ is the inverse of $w.\mathtt{count}$. The result of the binary search is not an artificial node.
\end{proof}

Let $[i..j]$ be the identifier of a node of $\mathcal{T}(G)$, and let $[i'..j']$ be the identifier of its weighted ancestor. Since it is easy to transform the node of $G$ that corresponds to $[i'..j']$ into interval $[i'..j']$ itself, Lemma \ref{lemma:DAG_sum_weights} effectively implements a map from $[i..j]$ to $[i'..j']$ in $O(\log{N})$ time.

Applying Lemma \ref{lemma:DAG_sum_weights} to $\CDAWG_T$ is all we need to support the additional operations in Table \ref{tab:suffixTree2} efficiently:

\begin{theorem}\label{thm:suffixTreeRepresentation}
Let $T \in [1..\sigma]^{n-1}\#$ be a string. There are two representations of $\ST_T$ that support the operations in Table \ref{tab:suffixTree2} and in Table \ref{tab:suffixTree} with the specified time and space complexities.
\end{theorem}
\begin{proof}
Operation $\mathtt{selectLeaf}(i)$ returns an identifier of the $i$-th leaf of $\ST_T$ in lexicographic order. Recall from Section \ref{sec:cpm2015} that we store in a variable $\gamma\mathtt{.right}$ the number of characters of the right extension implied by arc $\gamma$ of $\CDAWG_T$. Thus, the length of the suffix associated with a leaf of $\ST_T$ (or equivalently, the position of that leaf in right-to-left string order) is the sum of all weights in the source-to-sink path of $\CDAWG_T$ that corresponds to the leaf. Since the sum of such weights is a telescoping function, we use the data structures in Lemma~\ref{lemma:DAG_sum_weights}, built on these weights, to compute the value $s$ of the sum in $O(\log{n})$ time, and we return tuple $(v,s,i,i)$, where $v$ is the sink of $\CDAWG_T$. Returning $|T|-s+1$ instead is enough to implement $\SA_{T}[i]$. Since Lemma \ref{lemma:DAG_sum_weights} supports also the extraction of all values of a telescoping function inside a depth-first range of leaves $[i..j]$, implementing $\SA_{T}[i..j]$ is straightforward.

Operation $\mathtt{lca}(i,j)$ returns the identifier of the lowest common ancestor, in $\ST_T$, of the $i$-th and the $j$-th leaf in lexicographic order. We use Lemma \ref{lemma:DAG_sum_weights} to compute both the node $v$ of $\CDAWG_T$ that corresponds to such common ancestor, and its string depth $s$, returning tuple $(v,s,x,y)$, where the range $[x..y] \supseteq [i..j]$ of the lowest common ancestor is computed during the top-down traversal of $\CDAWG_T$ using the weighted ancestor data structure on $\tau(\CDAWG_T)$. A similar approach allows one to return $\LCP[i]$, and a slight variation of the approach used to compute $\SA_{T}[i..j]$ supports also $\LCP[i..j]$. Operation $\mathtt{depth}(\mathtt{id}(v))$ returns the depth of the node $v$ of $\ST_T$ whose identifier is $\mathtt{id}(v)$. Since $\mathtt{id}(v)$ contains the range $[i..j]$ of $v$ in $\BWT_T$, we can proceed as in operation $\mathtt{lca}(i,j)$, and return the length of the path that the search traversed from the source of $\CDAWG_T$ to the node of $\CDAWG_T$ that corresponds to $v$. Operation $\mathtt{leftmostLeaf}(\mathtt{id}(v))$ returns the identifier of the smallest leaf in lexicographic order in the subtree of $\ST_T$ rooted at node $v$. Let $\mathtt{id}(v)=(v',\ell,i,j)$, and let $W$ be the longest maximal repeat in the equivalence class of node $v'$. Then, $\mathtt{leftmostLeaf}(\mathtt{id}(v))=(w',\ell+v'.\mathtt{left},i,i)$, where $w'$ is the sink of $\CDAWG_T$, and $v'.\mathtt{left}$ is the string length of the path, in $\ST_T$, that goes from the node of $\ST_T$ with string label $W$ to its leftmost leaf. We store $v'.\mathtt{left}$ at every node $v'$ of the CDAWG. 
Operation $\mathtt{rightmostLeaf}$ can be handled symmetrically. 
Operation $\mathtt{stringAncestor}(\mathtt{id}(v),d)$ (respectively, $\mathtt{ancestor}(\mathtt{id}(v),d)$) returns the identifier of the highest ancestor of $v$ in $\ST_T$ whose string depth (respectively, depth) is at least $d$. This can be implemented with the weighted ancestor query provided by Lemma \ref{lemma:DAG_sum_weights}, where the weight of arc $\gamma$ of $\CDAWG_T$ is $\gamma.\mathtt{right}$ (respectively, one).

Finally, by Property \ref{property:revCDAWG}, we support access to the value of the inverse suffix array at string position $i$ by building the data structures of Lemma \ref{lemma:DAG_sum_weights} on the compacted $\REV{\CDAWG}_T$, with arc weights corresponding to offsets between nested BWT intervals, and with a weighted ancestor data structure on $\tau(\REV{\CDAWG}_T)$ based on offsets between string positions. Note that all arcs that end at the same node of the compacted $\REV{\CDAWG}_T$ have distinct weights. Then, we evaluate the sum of edge weights from the root of $\mathcal{T}(\REV{\CDAWG}_T)$ to its $i$-th leaf in depth-first order. Implementing $\ISA_{T}[i..j]$ is also straightforward, and $\PLCP[i]$ can be supported using $\ISA_{T}[i]$. Assume that, while building $\REV{\CDAWG}_T$, we keep the first character of the label of every arc of $\CDAWG_T$ that starts from the root, we propagate it during compaction, and we store it at the nodes as described in Lemma \ref{lemma:bille}. Then, since $\mathcal{T}(\REV{\CDAWG}_T)$ is a parse tree of $T$, we can also return $T[i]$ in $O(\log{n})$ time and $T[i..j]$ in $O(\log{n}+j-i)$ time. Since the compacted reversed CDAWG is a grammar for $T$, the time for extracting $T[i..j]$ can be reduced to $O(\log{n}+(j-i)/\log_{\sigma}n)$ by using the $\mathtt{access}$ query described in \cite{belazzougui2015access}.
\end{proof}

\begin{corollary} \label{cor:sa}
Given a string $T \in [1..\sigma]^{n-1}\#$, there is a representation of the suffix array of $T$, of the inverse suffix array of $T$, of the LCP array of $T$, of the permuted LCP array of $T$, and of $T$ itself, that takes $O(\newe_T)$ words of space, and that supports random access to any position in $O(\log{n})$ time.
\end{corollary}

Note that Corollary \ref{cor:sa} yields immediately a representation of the compressed suffix array of $T$ \cite{sadakane2007compressed} that takes $O(\newe_T)$ words of space.

\section{Extensions and conclusion}

Our data structures provide immediate support for a number of queries of common use in pattern matching, in addition to those listed in Tables \ref{tab:suffixTree2} and \ref{tab:suffixTree}. For example, recall that an \emph{internal pattern matching query} $(i,j)$ asks for all the $\mathtt{occ}$ starting positions of $T[i..j]$ inside a string $T$ of length $n$. We can support such query in $O(\log{n}+\mathtt{occ})$ time, by combining an inverse suffix array query, a string ancestor query, and the extraction strategy of Lemma \ref{lemma:DAG_sum_weights}. Similarly, combining an inverse suffix array query with a lowest common ancestor query and a string depth query, allows one to compute the longest common prefix between two given suffixes of $T$ in $O(\log{n})$ time. Along the same lines, operation $\mathtt{letter}(\mathtt{id}(v),i)$, which returns the $i$-th character of the label of node $v$ of the suffix tree, can be supported in $O(\log{n})$ time. 
We can also implement in constant time operation $\mathtt{deepestNode}(\mathtt{id}(v))$, which returns the identifier of the first node with largest depth (or string depth) in the subtree of the suffix tree rooted at $v$ \cite{navarro2014fully}. If we choose not to store the BWT intervals of the nodes of the CDAWG as in the second row of Tables \ref{tab:suffixTree2} and \ref{tab:suffixTree}, we can implement in $O(\log{n})$ time operation $\mathtt{suffixLink}(\mathtt{id}(v),i)$, which returns the identifier of the node of the suffix tree that is reachable from $v$ after taking $i$ suffix links. This can be done by computing $\mathtt{lca}(\mathtt{id}(u),\mathtt{id}(w))$, where $\mathtt{id}(v)=(v',k,a,b)$, $\mathtt{id}(u)=(z,e,x,x)$, $\mathtt{id}(w)=(z,f,y,y)$, $z$ is the sink of the CDAWG, $e=|T|-(\SA[a]+i)+1$, $f=|T|-(\SA[b]+i)+1$, $x=\ISA[\SA[a]+i]$ and $y=\ISA[\SA[b]+i]$. By using the representation described in~\cite{bille2015random}, we can also support in $O(\log{n})$ time operations like $\mathtt{preorderSelect}(i)$, $\mathtt{postorderSelect}(i)$, $\mathtt{preorderRank}(v)$, $\mathtt{postorderRank}(v)$, $\mathtt{treeLevelSuccessor}(v)$ and $\mathtt{treeLevelPredecessor}(v)$. However, some operations on the topology of the suffix tree are not yet implemented by our data structures (see e.g. \cite{navarro2014fully}): it would be interesting to know whether they can be supported efficiently within the same space budget.


Recall from Section \ref{sec:cpm2015} that our current representation of the suffix tree supports reading the label of an arc in $O(\log{\log{n}})$ time per character, using the RLBWT of $\REV{T}$. It would be interesting to know whether this bound can be improved, and whether the RLBWT of $\REV{T}$ can be dropped. Another question for further research is whether the ubiquitous $O(\log{n})$ term in Table \ref{tab:suffixTree2} can be reduced while keeping the same asymptotic space budget, or whether a lower bound makes it impossible, along the lines of \cite{verbin2013data}.

On the applied side, it is not yet clear whether there is a subset of our algorithms that is practically applicable, and whether it could achieve competitive tradeoffs with respect to state-of-the-art suffix tree representations for highly repetitive collections. It would also be interesting to try and use our data structures for tuning specific applications to repetitive strings in practice, like matching statistics and substring kernels. For example, it turns out that some weighting functions used in substring kernels are telescoping \cite{smola2003fast}. Since our data structures support matching statistics \cite{belazzougui2015composite}, and since the computation of some substring kernels can be mapped onto matching statistics \cite{smola2003fast}, we can compute some substring kernels between a fixed $T$ and a query string of length $m$ in $O(m\log{n})$ time, using a data structure that takes just $O(\newe_T)$ words of space.



\appendix


\bibliographystyle{plain}
\bibliography{paper}


\end{document}